\documentclass[runningheads]{llncs}
\usepackage{amsmath}
\usepackage{braket}
\usepackage{qcircuit}
\usepackage{graphicx}
\usepackage[T1]{fontenc}
\usepackage{graphicx}
\begin{document}
\title{From Symmetric Toeplitz Hamiltonians to Quantum Circuits}
%
%
\author{Rayan Trabelsi}
\authorrunning{Rayan Trabelsi}
%
\institute{Laboratoire Méthodes Formelles, Inria, 91190, Gif-sur-Yvette, France}
\maketitle              
\begin{abstract}
This work introduces a quantum circuit synthesis framework for simulating the unitary time evolution under a subclass of symmetric Toeplitz Hamiltonians by decomposing them into specific diagonal matrices $M_k$. These matrices are then classified, to achieve significant simplification, into, $M_k$ when $k$ is a power of two, and congruence classes with constant coefficients.
Finally, we construct the explicit quantum circuit for the one-dimensional discrete Poisson equation. \\
This research was conducted under the supervision of Benoit Valiron during a Master's internship.

\end{abstract}
\section{Introduction}

Solving large systems of linear equations of the form $Hx = b$ is a fundamental problem across many fields of science. Quantum computing offers a potential speedup for these problems through algorithms like the Harrow-Hassidim-Lloyd (HHL) algorithm \cite{1}. The core of the HHL algorithm relies on simulating the time evolution of a quantum system described by the Hamiltonian $H$ which is achieved by implementing the unitary operator $e^{-iHt}$. This approach is effective because the eigenvectors of the matrix $H$ are the same as those of $e^{-iHt}$, and their eigenvalues are directly related.
A particularly important subclass of these problems arises from the discretization of differential equations, whose associated Hamiltonians often take the form of banded symmetric Toeplitz matrices.

In this paper, under the supervision of Benoit Valiron, we introduce a structured framework for decomposing specific symmetric Toeplitz matrices and synthesizing quantum circuits to approximate their unitary evolution. Our approach proceeds in three main steps. First, we express the target matrix as a sum of simpler matrices $M_k$, each corresponding to a specific diagonal. Second, we develop circuit‐construction techniques for two key cases: \textbf{(i)} $M_k$ where $k$ is a power of two, and \textbf{(ii)} congruence classes where coefficients of $M_k$ are equal. The second case reveals a remarkable simplification, where the sum of these matrices reduces to a simple tensor product of Pauli operators. Third, we assemble the time evolution operator via a Trotter–Suzuki expansion, balancing approximation error against circuit depth.
To illustrate the practical utility of our framework, we apply it to the canonical one‐dimensional Poisson equation (equivalently, the discrete Laplacian).

\section{Decomposition and Circuit Synthesis of Symmetric Toeplitz Hamiltonians}

In this section, we study a special class of hermitian matrices : symmetric Toeplitz matrices, and we develop quantum circuits to synthesize its unitary evolution. We begin by formally defining symmetric Toeplitz and banded Hamiltonians. We then decompose these Hamiltonians into sums of simpler indexed matrices $M_k$, leveraging their tensor product structure. Finally, we derive explicit quantum circuits to implement $e^{-iM_kt}$.

\subsection{Definition}

A Toeplitz Hamiltonian $H$ is a matrix whose entries are constant along each diagonal:
\[
H_{i,j} = H_{i+1,j+1} = H_{i-1,j-1}
\quad \text{for all valid } i,j
\]
A symmetric Toeplitz Hamiltonian is a $n \times n$ matrix that satisfies $H_{i,j}=H_{j,i}$, so its entries depend only on $|i-j|$:
\[
H_{i,j} = a_{|i-j|}
\]
\[
H = \begin{pmatrix}
a_0      & a_1      & a_2      & \cdots   & a_{n-1} \\
a_1      & a_0      & a_1      & \ddots   & \vdots  \\
a_2      & a_1      & a_0      & \ddots   & a_2     \\
\vdots   & \ddots   & \ddots   & \ddots   & a_1     \\
a_{n-1}  & \cdots   & a_2      & a_1      & a_0
\end{pmatrix}
\]
A symmetric Toeplitz Hamiltonian is banded with bandwidth $b$ where \\ $0 \leq b < n-1$ if it satisfies the additional constraint:
\[
a_k = 0 \quad \forall \, k > b
\]

\subsection{Decomposition}

Let $H$ be an $N\times N$ symmetric Toeplitz matrix with $N = 2^n$. We decompose it into a sum of $k$-dependant matrices:
\[
H = \sum_{k=0}^{N-1}M_k
\]
where each $M_k$ is given by:
\[
M_k = a_k \sum_{i=0}^{N-1-k} (\ket{i}\bra{i+k} + \ket{i+k}\bra{i})
\]
\[
M_0=\begin{pmatrix}
a_0      & 0      & 0      & \cdots   & 0 \\
0      & a_0      & 0      & \ddots   & \vdots  \\
0      & 0      & a_0      & \ddots   & 0     \\
\vdots   & \ddots   & \ddots   & \ddots   & 0     \\
0  & \cdots   & 0      & 0      & a_0
\end{pmatrix},
M_1 = 
\begin{pmatrix}
0      & a_1      & 0      & \cdots   & 0 \\
a_1      & 0      & a_1      & \ddots   & \vdots  \\
0      & a_1      & 0      & \ddots   & 0     \\
\vdots   & \ddots   & \ddots   & \ddots   & a_1     \\
0  & \cdots   & 0      & a_1      & 0
\end{pmatrix},
\
M_2 = 
\begin{pmatrix}
0      & 0      & a_2      & \cdots   & 0 \\
0      & 0      & 0      & \ddots   & \vdots  \\
a_2      & 0      & 0      & \ddots   & a_2     \\
\vdots   & \ddots   & \ddots   & \ddots   & 0     \\
0  & \cdots   & a_2      & 0      & 0
\end{pmatrix} \quad \cdots
\]

\begin{theorem}
\label{theorem:1}
If $k$ is a power of two, $k=2^m$, the structure of $M_k$ is given by: 
\[
M_k = 
\begin{cases} 
a_k \left( T_k + P^k \, E_k \, P^{-k} \right) & \text{if } m < n-1 \\ 
a_k T_k & \text{if } m = n-1 
\end{cases} 
\]
where $T_k = I^{\otimes (n-m-1)} \otimes X \otimes I^{\otimes m}$, $P$ is a permutation matrix $(P^{-1} = P^T)$ and a cyclic decrement operator defined by $P\ket{i} = \ket{(i-1) \pmod N}$, and $E_k$ is the $N \times N$ matrix obtained from $T_k$ by setting its bottom-right $2k \times 2k$ submatrix to the zero matrix. More precisely,
\[
(E_k)_{i,j} = 
\begin{cases} 
(T_k)_{i,j} & \text{if } i \leq 2^n-2k \;\text{ or }\; j \leq 2^n-2k \\ 
0 & \text{if } i > 2^n-2k \;\text{ and }\; j >2^n-2k
\end{cases} 
\]
\end{theorem}

\begin{proof} \textbf{Case 1: $m = n-1$.} In this case, $k=2^{n-1} = N/2$. The matrix $M_k$ becomes
\[
M_{N/2} = a_{N/2} \sum_{i=0}^{N/2-1} (\ket{i}\bra{i+N/2} + \ket{i+N/2}\bra{i}) 
\]
Consider the operator $T_{N/2} = X \otimes I^{\otimes(n-1)}$. This operator flips the most significant qubit of a basis state $\ket{i} = \ket{i_{n} \dots i_1}$. Its action is to map $\ket{i} \leftrightarrow \ket{i \pm N/2}$. For an index $i \in [0, N/2-1]$, its most significant bit $i_{n}$ is 0, so $T_{N/2}\ket{i} = \ket{i+N/2}$. For an index $j \in [N/2, N-1]$, its most significant bit $j_{n}$ is 1, so $T_{N/2}\ket{j} = \ket{j-N/2}$. Thus, the non-zero matrix elements of $T_{N/2}$ are precisely $\bra{l} T_{N/2} \ket{l + N/2} = 1 = \bra{l + N/2} T_{N/2} \ket{l}$ for $l = 0, \dots, N/2 - 1$. This matches the structure of $M_{N/2}$ up to the scalar $a_{N/2}$, so $M_{N/2} = a_{N/2} T_{N/2}$. \\ 
 
\textbf{Case 2: $m < n-1$.} We partition the sum defining $M_k$ according to the parity of $\lfloor i / k \rfloor$:
\[
M_k = M_k^{\text{even}} + M_k^{\text{odd}}
\]
where
\[
M_k^{\text{even}} = a_k \sum_{\substack{i=0 \\ \lfloor i/k \rfloor \text{ even}}}^{N-1-k} \left( \ket{i}\bra{i+k} + \ket{i+k}\bra{i} \right)
\]
\[
M_k^{\text{odd}} = a_k \sum_{\substack{i=0 \\ \lfloor i/k \rfloor \text{ odd}}}^{N-1-k} \left( \ket{i}\bra{i+k} + \ket{i+k}\bra{i} \right)
\]

The operator $T_k = I^{\otimes (n-m-1)} \otimes X \otimes I^{\otimes m}$ flips the bit at position $m$ (counting from the least significant bit as position 0) in the binary representation of basis indices. The parity of $\lfloor i / 2^m \rfloor=\lfloor i / k \rfloor$ is even if and only if the $m$-th bit $b_m(i) = 0$. \\
When $\lfloor i / k \rfloor$ is even, we have $T_k\ket{i} = \ket{i+k}$ and $T_k\ket{i+k}=\ket{i}$.

Thus, $T_k$ has non-zero matrix elements $\bra{i} T_k \ket{i + k} = 1 = \bra{i + k} T_k \ket{i}$ for all $i$ with $b_m(i) = 0$ and $i + k < N$ (since $i \leq N - 1 - k$) with no other nonzero elements connecting the pair $\{i,i+k\}$. This exactly matches $M_k^{\text{even}} / a_k$, so $a_k T_k = M_k^{\text{even}}$.

For the odd part, $E_k$ is $T_k$ except in the bottom-right $2k \times 2k$ submatrix, which is set to zero. This removes the matrix elements of $T_k$ corresponding to pairs $\{i, i + k\}$ if $b_m(i) = 0$ where $i \geq N - 2k$.

\noindent Conjugating by $P^k$ gives $P^k E_k P^{-k}$, whose matrix elements are $\bra{j} P^k E_k P^{-k} \ket{l} = \bra{j + k \bmod N} E_k \ket{l + k \bmod N}$. This shifts the retained pairs of $E_k$ (those with $b_m(i) = 0$ and $i < N - 2k$) to pairs $\{i + k, i + 2k\}$, where now $b_m(i + k) = 1$ (so $\lfloor (i + k)/k \rfloor$ is odd) and $k \leq i + k < N - k$. Due to the truncation in $E_k$, no wrap-around occurs, and the shifted pairs exactly match those in $M_k^{\text{odd}} / a_k$. Thus, $a_k P^k E_k P^{-k} = M_k^{\text{odd}}$.
\end{proof}

\begin{theorem}
For any fixed integer $j$ such that $1 \le j < n$, consider the congruence class of indices $C_j = \{k \mid 1 \le k < 2^n, k \equiv 2^{j-1} \pmod{2^j}\}$. If the coefficients $a_k$ are constant for all $k \in C_j$, then the sum of the corresponding $M_k$ matrices simplifies to a sum of tensor products of Pauli $X$ operators: \[ \sum_{k \in C_j} M_k = a_{2^{j-1}} \sum_{k \in C_j} \bigotimes_{i=1}^n X^{b_i(k)} \] where $b_i(k)$ is the $i$-th bit in the $n$-bit binary representation of $k$. 
\end{theorem} 

\begin{proof} The proof relies on the observation that the permutation-dependent terms in the decomposition of individual $M_k$ cancel out when summed over a complete congruence class. Any $M_k$ can be conceptually decomposed into a primary tensor product term which we denote $T_k' = \bigotimes_{i=1}^n X^{b_i(k)}$, and a sum of correction terms $\mathcal{C}_k$ involving permutation matrices. We can express this as: 
\[ 
M_k = a_k (T_k' + \mathcal{C}_k) 
\]
We consider the sum over all $k \in C_j$. The condition $k \equiv 2^{j-1} \pmod{2^j}$ fixes the $j$ least significant bits of $k$ to $(10...0)_2$, while the $n-j$ most significant bits remain free to vary over all $2^{n-j}$ possibilities.
Given the constant coefficient $a_k = a_{2^{j-1}}$ for all $k \in C_j$, the sum becomes: 
\[
\sum_{k \in C_j} M_k = a_{2^{j-1}} \left( \sum_{k \in C_j} T_k' + \sum_{k \in C_j} \mathcal{C}_k \right) 
\]
The crucial insight is that the sum of the correction terms over a complete congruence class vanishes. 
\[
\sum_{k \in C_j} \mathcal{C}_k = 0 
\]
This cancellation occurs because the correction terms are structured to precisely offset each other when all combinations of the $n-j$ most significant bits are summed over. For every missing element from one's $M_k$, there is a corresponding element in excess from another that cancels it. In other words, it ensures that for every introduced $\mathcal{C}_k$, a corresponding $\mathcal{C}_k$ or sum of $\mathcal{C}_k$s with opposite sign is introduced by another $M_k$ in the sum\footnote{In order to form $T_k'$, each $M_k$ takes/gives the elements $M_k-T_k'$ from/to other $M_k$. Since the elements are all equal over a same congruence class, it leads to a cancellation of $\sum_{k \in C_j} \mathcal{C}_k$ }. The net effect is a perfect cancellation. This is a structural property of symmetric Toeplitz matrices. 
With the correction terms eliminated, the sum simplifies to: 
\[
\sum_{k \in C_j} M_k = a_{2^{j-1}} \sum_{k \in C_j} T_k' = a_{2^{j-1}} \sum_{k \in C_j} \bigotimes_{i=1}^n X^{b_i(k)} 
\]
\end{proof}

\subsection{Circuit Synthesis} 
Having established the decomposition of the Hamiltonian $H$ into $M_k$, we now detail the synthesis of the corresponding quantum circuits for the unitary evolution $e^{-iHt}$. 
When the Hamiltonian is expressed as a sum of non-commuting terms, $H = \sum_k M_k$, the overall evolution is approximated using a first-order Trotter-Suzuki decomposition: 
\[ 
e^{-iHt} = e^{-i(\sum_k M_k)t} \approx \left( \prod_k e^{-iM_k t/u} \right)^u 
\] 
where $u$ is the Trotter number. The accuracy of the approximation increases with $u$. We present the specific circuit constructions for $M_k$ for the two cases derived from our theorems. 

\subsubsection*{Circuit Synthesis from Theorem 1:} 
From our previous theorem, for $k=2^m$
\[
M_k = 
\begin{cases} 
a_k \left( T_k + P^k \, E_k \, P^{-k} \right) & \text{if } m < n-1 \\ 
a_k T_k & \text{if } m = n-1 
\end{cases} 
\]  
Since $[T_k, P^k \, E_k \, P^{-k}] \neq 0$, we use Trotterization again to simulate the evolution when $m<n-1$: 
\[ 
U_k(t) = e^{-itM_k} \approx \left( e^{-i\frac{t}{v}a_k T_k} e^{-i\frac{t}{v}a_kP^k  E_k  P^{-k}} \right)^v 
\] 
The implementation of one Trotter step thus requires constructing circuits for two unitaries. The first unitary acts only on the $(m+1)$-th qubit. Its circuit is a single rotation gate: 
\[ 
e^{-i\theta_vT_k} = I^{\otimes(n-m-1)} \otimes e^{-i\theta_v X} \otimes I^{\otimes m} = I^{\otimes(n-m-1)} \otimes R_x(2\theta_v) \otimes I^{\otimes m} 
\] where $\theta_v = \frac{t}{v}a_k$. \\
The second unitary is implemented by conjugating the evolution of $E_k$ with the permutation operator $P^k$: 
\[ 
e^{-i\theta_vP^k  E_k  P^{-k}} = e^{-i\theta_v P^k E_k P^{-k}} = P^k e^{-i\theta_v E_k} P^{-k} 
\] 
The exponential of $E_k$ can be implemented with a set of controlled rotations. The operators $P^k$ and $P^{-k}$ correspond to quantum subtractors/adders acting on the register of qubits. \\
A complete circuit for one Trotter step is shown in Figure~\ref{fig:1}. The complexity of this step is dominated by the implementation of the quantum adder, which can be realized with $O( \log (n-m))$ depth using a carry-lookahead design at the cost of additional ancilla qubits.\cite{2} 
\begin{figure}
\[
\Qcircuit @C=1em @R=0.8em {
  \lstick{q_1}     & \qw & \multigate{8}{P^{k}} & \qw & \qw & \qw          & \qw                & \qw                & \multigate{8}{P^{-k}}  & \qw \\
  \lstick{q_2}     & \qw & \ghost{P^{k}}        & \qw & \qw                & \qw                & \qw                & \qw       & \ghost{P^{-k}}        & \qw \\
  \vdots           &  &                   &     &                    &                    &                    &           &                       &     \\
                   &  &                    &     &                    &                    &                    &           &                       &     \\
  \lstick{q_{m+1}} & \gate{R_x(2\theta_1)} & \ghost{P^{k}}        & \qw &\gate{R_x(2\theta_1)}  & \targ              & \gate{R_x(-2\theta_1)} & \targ    & \ghost{P^{-k}}        & \qw \\
  \lstick{q_{m+2}} & \qw & \ghost{P^{k}}        & \qw & \qw                & \ctrl{-1}          & \qw                & \ctrl{-1} & \ghost{P^{-k}}        & \qw \\
  \vdots           &  &                    &     &                    & \vdots             &                    & \vdots    &                       &     \\
                   &   &                   &     &                    &                    &                    &           &                       &     \\
  \lstick{q_n}     & \qw & \ghost{P^{k}}        & \qw & \qw                & \ctrl{-1}          & \qw                & \ctrl{-1} & \ghost{P^{-k}}        & \qw
}
\] 
\caption{Quantum circuit for $v=1$ under $M_k$ for $k=2^m$ with $m<n-1$.}
\label{fig:1}
\end{figure}
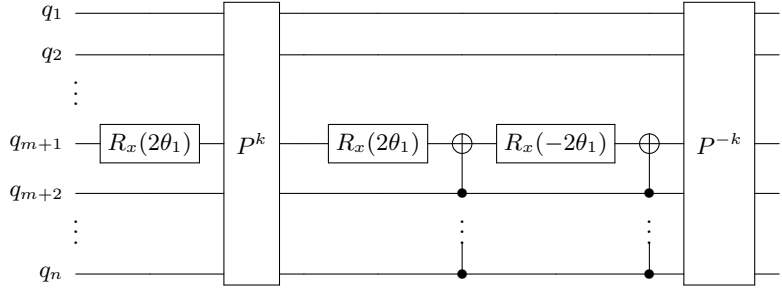 \\
For $m=n-1$, we can just assume that the second unitary used above is the identity. Thus, the circuit is simply an $R_x(2a_{2^{n-1}}t)$ applied on the last qubit $q_n$.

\subsubsection*{Circuit Synthesis from Theorem 2:} When the coefficients $a_k$ are constant over a congruence class $C_j$, the Hamiltonian term $\sigma_j = \sum_{k \in C_j} M_k$ simplifies significantly to a sum whose terms commute. The unitary evolution $U_{\sigma_j}(t) = e^{-it\sigma_j}$ can therefore be implemented exactly without Trotterization. \\
The binary constraint \(k \equiv 2^{j-1} \pmod{2^j}\) fixes the \(j\) least significant bits to \((b_j\dots b_1) = (10\dots0)\). Thus, the general term in the sum is:
\[
X^{b_n} \otimes X^{b_{n-1}} \otimes \cdots \otimes X^{b_{j+1}} \otimes X \otimes I^{\otimes (j-1)}
\]

Consider the sum over all possible combinations of $b_{j+1}, \dots, b_n$.
This is equivalent to summing over all possible tensor products where the last element is $X \otimes I^{\otimes (j-1)}$. Thus, $\sigma_j$ can be represented as :
$$ \sigma_j \;= a_{2^{j-1}} \sum_{b_n \in \{0,1\}} \dots \sum_{b_{j+1} \in \{0,1\}} \left(X^{b_n} \otimes X^{b_{n-1}} \otimes \dots \otimes X^{b_{j+1}} \otimes X \otimes I^{\otimes (j-1)}\right) $$
This sum can be factored due to the linearity of the tensor product:
$$ \sigma_j \;= a_{2^{j-1}} \left(\sum_{b_n \in \{0,1\}} X^{b_n}\right) \otimes \left(\sum_{b_{n-1} \in \{0,1\}} X^{b_{n-1}}\right) \otimes \dots \otimes \left(\sum_{b_{j+1} \in \{0,1\}} X^{b_{j+1}}\right) \otimes X \otimes I^{\otimes (j-1)} $$
Since $\sum_{b \in \{0,1\}} X^b = X^0 + X^1 = I + X$, it simplifies to:
$$ \sigma_j = a_{2^{j-1}} \; \underbrace{(I + X) \otimes \dots \otimes (I + X)}_{n-j \text{ times}} \otimes X \otimes I^{\otimes (j-1)} = a_{2^{j-1}}(I+X)^{\otimes(n-j)} \otimes X \otimes I^{\otimes(j-1)} $$ 

Using the identity $X=HZH$ and the fact that $I+X = 2H\ket{0}\bra{0}H$, we can conjugate $\sigma_j$ into a simpler diagonal form: 
\[
\sigma_j = a_{2^{j-1}} 2^{n-j} (H^{\otimes(n-j+1)} \otimes I^{\otimes(j-1)}) \left( (\ket{0}\bra{0})^{\otimes(n-j)} \otimes Z \otimes I^{\otimes(j-1)} \right) (H^{\otimes(n-j+1)} \otimes I^{\otimes(j-1)})
\]
Since the inner operator is diagonal in the computational basis, the evolution $U_{\sigma_j}(t) = e^{-it\sigma_j}$ becomes: 
\[ 
U_{\sigma_j}(t) = \mathcal{H} \; e^{ -it a_{2^{j-1}} 2^{n-j}  (\ket{0}\bra{0})^{\otimes(n-j)} \otimes Z \otimes I^{\otimes(j-1)} } \; \mathcal{H} 
\] 
where $\mathcal{H} = H^{\otimes(n-j+1)} \otimes I^{\otimes(j-1)}$. The circuit associated to the unitary $U_{\sigma_j}$ is given in Figure~\ref{fig:2} where $\phi_j = 2 t  a_{2^{j-1}} 2^{n-j} = t a_{2^{j-1}} 2^{n-j+1} $ 
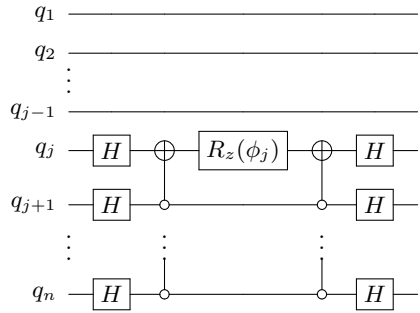
\begin{figure}
\[
\Qcircuit @C=1em @R=0.8em {
  \lstick{q_1}     & \qw      & \qw        & \qw                & \qw         & \qw & \qw \\
   & & & & & & & \\
  \lstick{q_2}     & \qw      & \qw        & \qw                & \qw         & \qw & \qw \\
    \vdots               &          & &                    &       &          &     \\
    & & & & & & & \\
  \lstick{q_{j-1}} & \qw      & \qw        & \qw                & \qw         & \qw & \qw \\
  \lstick{q_{j}}   & \gate{H} & \targ      & \gate{R_z(\phi_j)} & \targ       & \gate{H} & \qw \\
  \lstick{q_{j+1}} & \gate{H} & \ctrlo{-1} & \qw                & \ctrlo{-1}  & \gate{H}      & \qw \\
  \vdots           &          &  \vdots          &                    &  \vdots           &          &     \\
                   &          &            &                    &            &          &     \\
  \lstick{q_n}     & \gate{H} & \ctrlo{-1} & \qw                & \ctrlo{-1}  & \gate{H}      & \qw
}
\]
\caption{Quantum circuit of a sum over congruence class $C_j$}
\label{fig:2}
\end{figure}

\section{Poisson equation}
In one dimension, the continuous Poisson equation
\[
-\frac{d^2 u(x)}{dx^2} = f(x), 
\quad x\in[0,L],\quad u(0)=u(L)=0
\]
is commonly discretized \cite{4} on a uniform grid \(x_i = i\Delta x\) (\(i=0,\dots,N\)), giving the tridiagonal system
\[
\mathcal{L}\,\mathbf{u} = \mathbf{f}, 
\quad
\mathcal{L} = \frac{1}{\Delta x^2}
\begin{pmatrix}
2 & -1 &        &        \\
-1& 2  & -1     &        \\
  & \ddots &\ddots&\ddots\\
  &        & -1   & 2
\end{pmatrix}
\]
which is a symmetric banded Toeplitz matrix with bandwidth $b=1$.\\
The unitary evolution $e^{-i\mathcal{L}t}$ is equal to $e^{-i\mathcal{L}_0t}e^{-i\mathcal{L}_1t}$ since $[\mathcal{L}_0,\mathcal{L}_1]=0$ with 
\[
\mathcal{L}_0 = \frac{1}{\Delta x^2}
\begin{pmatrix}
2 & 0 &        &        \\
0& 2  & 0     &        \\
  & \ddots &\ddots&\ddots\\
  &        & 0   & 2
\end{pmatrix},
\quad
\mathcal{L}_1 = \frac{-1}{\Delta x^2}
\begin{pmatrix}
0 & 1 &        &        \\
1& 0  & 1     &        \\
  & \ddots &\ddots&\ddots\\
  &        & 1   & 0
\end{pmatrix}
\]
$e^{-i\mathcal{L}_0t}$ corresponds to a global phase, which has no effect on the circuit and can therefore be omitted in the implementation. \\
$\mathcal{L}_1$ is the matrix associated to $k=1 = 2^0$. Thus, as a result of Theorem \ref{theorem:1}
\[
\mathcal{L}_1 = \frac{-1}{\Delta x^2} \,(\, I^{\otimes (n-1)} \, \otimes \, X \,+\, P \, R_1 \, P^{-1} \,)
\]
Its unitary is then given by
\[
e^{-i\mathcal{L}_1t} \approx  \big( \; I^{\otimes (n-1)} \otimes e^{\frac{it}{v\Delta x^2}X} \; \; P \, e^{\frac{it}{v\Delta x^2}R_1} \, P^{-1} \; \big)^v
\]

\noindent Therefore, the circuit is obtained by repeating the following subroutine $v$ times where $v$ is the Trotter step: 
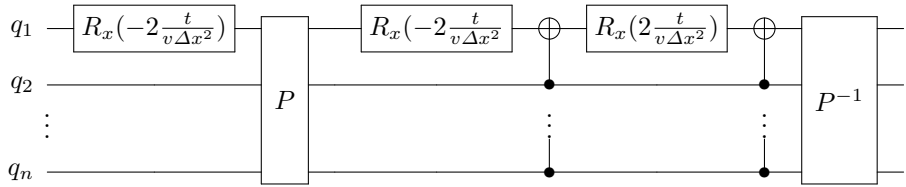
\[
\Qcircuit @C=1em @R=0.8em {
  \lstick{q_{1}}   & \gate{R_x(-2\frac{t}{v\Delta x^2})} & \multigate{4}{P}        & \qw &\gate{R_x(-2\frac{t}{v\Delta x^2}}  & \targ              & \gate{R_x(2\frac{t}{v\Delta x^2})} & \targ     & \multigate{4}{P^{-1}}        & \qw \\
  \lstick{q_{2}}   & \qw                                 & \ghost{P}           & \qw & \qw                          & \ctrl{-1}          & \qw                           & \ctrl{-1} & \ghost{P^{-1}}               & \qw \\
  \vdots           &                                     &                         &     &                              & \vdots             &                               & \vdots    &                              &     \\
                   &                                     &                         &     &                              &                    &                               &           &                              &     \\
  \lstick{q_n}     & \qw                                 & \ghost{P}           & \qw & \qw                          & \ctrl{-1}          & \qw                           & \ctrl{-1} & \ghost{P^{-1}}               & \qw
}
\]

\section{Conclusion}
We presented a framework for synthesizing quantum circuits to simulate time evolution under specific symmetric Toeplitz Hamiltonians via two key decomposition theorems. By breaking any such Hamiltonian into $M_k$ matrices, we have shown that this leads to generalized circuits, for $M_k$ when $k$ is a power of two, and congruence classes with constant coefficients.

The practical relevance of this approach was demonstrated by constructing a circuit for the 1D discrete Poisson equation. Future work will focus on extending this framework to general symmetric Toeplitz matrices and performing detailed benchmarking of the algorithm to analyze Trotter errors and resource overhead.

%
%
%
%

\end{document}